\documentclass[12pt]{article}

\usepackage{amssymb,amsmath,amsthm,amscd,mathrsfs,amsbsy,amsfonts}
\usepackage{pstricks,pst-node}
\usepackage{amsbsy,young,colortbl,cite}

\usepackage{graphicx}

\usepackage{amsmath}

\usepackage{amssymb,amsthm}
\usepackage[english]{babel}
\newtheorem{proposition}{Proposition}
\newtheorem{lemma}{Lemma}
\newtheorem{definition}{Definition}

\def\({\left(}
\def\){\right)}
\def\[{\begin{eqnarray}}
\def\]{\end{eqnarray}}

\usepackage{amsmath}
\numberwithin{equation}{section}

\begin{document}

\title{Extensions of the finite nonperiodic Toda lattices with indefinite metrics
}

\author{
\  \  Jian Li, Chuanzhong Li\footnote{Corresponding author:lichuanzhong@nbu.edu.cn.}\\
\small School of Mathematics and Statistics,  Ningbo University, Ningbo, 315211, China}

\date{}

\maketitle

\abstract{
In this paper, we firstly construct a weakly coupled Toda lattices with indefinite metrics which consist of $2N$ different coupled Hamiltonian systems. Afterwards, we consider the iso-spectral manifolds of extended tridiagonal Hessenberg matrix with indefinite metrics  what is an extension of a strict tridiagonal matrix with indefinite metrics. For the initial value problem of the extended symmetric Toda hierarchy with indefinite metrics, we introduce the inverse scattering procedure in terms of eigenvalues by using the Kodama's method. In this article,  according to the orthogonalization procedure of Szeg\"{o}, the relationship between the $\tau$-function and the given Lax matrix is also discussed. We can verify the results derived from the orthogonalization procedure with a simple example. After that, we  construct a strongly coupled Toda lattices with indefinite metrics  and derive its tau structures. At last, we generalize the weakly coupled Toda lattices with indefinite metrics  to the  $Z_{n}$-Toda lattices with indefinite metrics.
}\\

\allowdisplaybreaks
\tableofcontents

\section {Introduction}

\ \ \ \ Toda system is one of the most important integrable systems in mathematical physics. Many mathematicians have made significant contributions to the Toda equations and its generalization \cite{16}, such as M. Toda, Y. Kodama, etc. In recent years, some senior mathematicians have used different methods to study the Toda equations, such as symmetry and bilinear method, etc. This paper aims to extend the Toda equations via an extended algebra group{\cite{1}}, and the solutions can be pasted together to constitute a compact manifold.
On the basis of \cite{18,19}, the Toda lattices are produced by semisimple Lie algebra.
In the process of solving the extended Toda equations, we use the inverse scattering method, it also promotes the development of mathematical physics, integrable systems and Lie algebras \cite{2,14}. According to the Hamiltonian systems of $2N$ particles which described by the finite nonperiodic Toda lattice hierarchy \cite{7}, then we introduce a pair of Hamiltonian systems $(H, \widehat{H})$  given by
\begin{equation}\label{a2}
\begin{cases}
H=\frac{1}{2}\sum^{N}_{i=1}y^{2}_{i}+\sum^{N}_{i=1}\exp(x_{i}-x_{i+1}),\\
\widehat{H}=\sum^{N}_{i=1}y_{i}\widehat{y}_{i}+\sum^{N}_{i=1}(\widehat{x}_{i}-\widehat{x}_{i+1})\exp(x_{i}-x_{i+1}).
\end{cases}
\end{equation}
    The topology of an iso-spectral set of tridiagonal Hessenberg matrices was considered{\cite{7}}, and it has distinct real eigenvalues in the following form,
\begin{equation}\label{a3}
L_{H}=\left(\begin{smallmatrix}
\alpha_{1}&1&0&0&\cdots&0&0&0\\
\beta_{1}&\alpha_{2}&1&0&\cdots&0&0&0\\
 & &\ddots&\ddots&\ddots& & & \\
0&0&0&0&\cdots&\cdots&\alpha_{N-1}&1\\
0&0&0&0&\cdots&\cdots&\beta_{N-1}&\alpha_{N}\\
\end{smallmatrix}
\right)\\,
\end{equation}
\begin{equation}\label{a31}
\widehat{L}_{H}=\left(\begin{smallmatrix}
\widehat{\alpha}_{1}&1&0&0&\cdots&0&0&0 \\
\widehat{\beta}_{1}&\widehat{\alpha}_{2}&1&0&\cdots&0&0&0\\
 & &\ddots&\ddots&\ddots& & &\\
0&0&0&0&\cdots&\cdots&\widehat{\alpha}_{N-1}&1\\
0&0&0&0&\cdots&\cdots&\widehat{\beta}_{N-1}&\widehat{\alpha}_{N}\\
\end{smallmatrix}
\right),
\end{equation}
where the variables $\widehat{\alpha}_{k}$ and $\widehat{\beta}_{k}$ in $L_{H}$ and $\widehat{L}_{H}$ are expressed as $s_{k}\widehat{a}_{k}=-\frac{1}{2}\widehat{y}_{k}$ and $\widehat{b}_{k}=\frac{{\widehat{x}_{k}}-{\widehat{x}_{k+1}}}{4}\exp(\frac{x_{k}-x_{k+1}}{2})$ with $s_{i}=\pm1$. while different signs of $s_{i}$ may create different systems. The initial value problem of Toda equations were studied by applying the inverse scattering method in \cite{13}, we generalize the above method and get general results on the basis of the method mentioned in the references. With the help of the defined inner product in this paper, the elements of $L$ and $\widehat{L}$ can be expressed in a simple way.

This paper is arranged as follows.  In Section 2, we construct the weakly coupled Toda equations via the transformation \cite{6}, by calculating the equations \eqref{29}, the elements of $L$ and $\widehat{L}$ can be expressed through an inner product and the  initial value of the weakly coupled Toda lattices. Finally, we briefly describe the relationships between the elements of $L_{H}$, $\widehat{L}_{H}$ and $\tau$-functions, and analyse specific relationships between $\widehat{\tau}_{i}$ and $\widehat{D}_{i}$. In section 3, we give a proof that the wave functions of the weakly coupled Toda lattices  can be solved by the inverse scattering method with the Gram-Schmidt's orthogonalization. In Section 4, we illustrate these results with a specific example, and some properties of the elements in the example are discussed. In section 5, we introduce the strongly coupled Toda lattices with indefinite metrics, and give some different conclusions compare with the weakly coupled Toda lattices. In section 6, first we give a definition of the $Z_{n}$-Toda {equations} by the algebraic transformation, and the {solutions of  the $Z_{n}$-Toda equations} are obtained according to the initial value.

\section{Weakly coupled Toda lattices with indefinite metrics}
 \ \ \ \ In this section, we define a weakly coupled Toda lattices with indefinite metrics. For the Hamiltonian \eqref{a2}, a  transformation of variables will be introduced similarly as the one from Flaschka \cite{6}, which is about the classical Toda lattices with indefinite metrics:
\begin{equation}\label{2.1}
\begin{cases}
s_{k}a_{k}=-\frac{y_{k}}{2},\\
s_{k}\widehat{a}_{k}=-\frac{\widehat{y}_{k}}{2},k=1,\ldots, N;
\end{cases}
\end{equation}
\begin{equation}\label{2.2}
\begin{cases}
b_{k}=\frac{1}{2}\exp(\frac{x_{k}-x_{k+1}}{2}),\\
\widehat{b}_{k}=\frac{{\widehat{x}_{k}}-{\widehat{x}_{k+1}}}{4}\exp(\frac{x_{k}-x_{k+1}}{2}), k=1,\ldots, N-1.
\end{cases}
\end{equation}
{Then} the extended Toda equations are written in this form with $b_{0}=\widehat{b}_{0}=b_{N}=\widehat{b}_{N}=0$,
\begin{equation}\label{2.3}
\begin{cases}
\frac{da_{k}}{dt}=\frac{1}{2}(s_{k+1}b_{k}^{2}-s_{k-1}b_{k-1}^{2}),\\
\frac{d\widehat{a}_{k}}{dt}=s_{k+1}b_{k}\widehat{b}_{k}-s_{k-1}b_{k-1}\widehat{b}_{k-1};
\end{cases}
\end{equation}
\begin{equation}\label{2.4}
\begin{cases}
\frac{db_{k}}{dt}=\frac{1}{4}b_{k}(s_{k+1}a_{k+1}-s_{k}a_{k}),\\
\frac{d\widehat{b}_{k}}{dt}=\frac{1}{4}[(s_{k}\widehat{a}_{k}-s_{k+1}\widehat{a}_{k+1})b_{k}+(s_{k}a_{k}-s_{k+1}a_{k+1})\widehat{b}_{k}].
\end{cases}
\end{equation}
The equations \eqref{2.3} and \eqref{2.4} can also be expressed as the following Lax equations:
\begin{equation}\label{2.5}
\begin{cases}
\frac{d}{dt}L=[B,L],\\
\frac{d}{dt}\widehat{L}=[\widehat{B},L]+[B,\widehat{L}],
\end{cases}
\end{equation}
where $L$ and $\widehat{L}$ are a $N\times N$ tridiagonal matrix with real entries,
\begin{equation}\label{2.6}
L=\left(\begin{smallmatrix}
s_{1}a_{1}&s_{2}b_{1}&\cdots&0\\
s_{1}b_{1}&s_{2}a_{2}&\vdots&0\\
 &\ddots&\ddots&\\
0&\cdots&s_{N-1}a_{N-1}&s_{N}b_{N-1}\\
0&\cdots&s_{N-1}b_{N-1}&s_{N}a_{N}\\
\end{smallmatrix}
\right)\\,
\end{equation}
\begin{equation}\label{9.0}
\widehat{L}=\left(\begin{smallmatrix}
s_{1}\widehat{a_{1}}&s_{2}\widehat{b_{1}}&\cdots&0\\
s_{1}\widehat{b_{1}}&s_{2}\widehat{a_{2}}&\cdots&0\\
 &\ddots&\ddots&\\
0&\cdots&s_{N-1}\widehat{a}_{N-1}&s_{N}\widehat{b}_{N-1}\\
0&\cdots&s_{N-1}\widehat{b}_{N-1}&s_{N}\widehat{a}_{N}\\
\end{smallmatrix}
\right)\\,
\end{equation}
$B$ and $\widehat{B}$ are the projection of $L$, $\widehat{L}$ given by
\begin{equation}\label{27}
\begin{cases}
B:=\frac{1}{4}[(L)_{>0}-(L)_{<0}],\\
\widehat{B}:=\frac{1}{4}[(\widehat{L})_{>0}-(\widehat{L})_{<0}].
\end{cases}
\end{equation}
Note that, $LS^{-1}$ and $\widehat{L}S^{-1}$  are symmetric tridiagonal matrix and $S$ is a diagonal matrix $S=\mathrm{diag}(s_{1}, s_{2},..., s_{N})$.
 In the this section, the extended Hamilton equations \eqref{a2} can be expressed by Lax equations \eqref{2.5} and the matrices \eqref{a3}, \eqref{a31}. In fact, the variables in \eqref{a3} and \eqref{a31} are given by
\begin{equation}\label{28}
\begin{cases}
\alpha_{k}=s_{k}a_{k},\\
\widehat{\alpha}_{k}=s_{k}\widehat{a}_{k},\\
\beta_{k}=s_{k}s_{k+1}b_{k}^{2},\\
\widehat{\beta}_{k}=2s_{k}s_{k+1}b_{k}\widehat{b}_{k},
\end{cases}
\end{equation}
and there is no doubt that they are equivalent. In order to solve the problem of the Lax equations \eqref{2.5}, we can use the inverse scattering method to construct a specific formula {from \cite{13}}. There are four linear equations that are contained in \eqref{2.5},
\begin{equation}\label{29}
\begin{cases}
L\Phi=\Phi\Lambda,\\
\widehat{L}\Phi+L\widehat{\Phi}=\widehat{\Phi}\Lambda,\\
\frac{d}{dt}\Phi=B\Phi,\\
\frac{d}{dt}\widehat{\Phi}=\widehat{B}\Phi+B\widehat{\Phi},
\end{cases}
\end{equation}
where $\Phi$ is the eigenmatrix of $L$, and $\widehat{\Phi}$ is the eigenmatrix of $\widehat{L}$, $\Lambda=\mathrm{diag}(\lambda_{1},..., \lambda_{N-1}, \lambda_{N})$ is a diagonal matrix, and $\Phi$, $\widehat{\Phi}$ also satisfy the following relationship:
\begin{equation}\label{210}
\begin{cases}
\Phi^{T}S\Phi=S,\\
\Phi^{T}S\widehat{\Phi}+\widehat{\Phi}^{T}S\Phi=0,\\
\Phi S^{-1}\Phi^{T}=S^{-1},\\
\Phi S^{-1}\widehat{\Phi}^{T}+\widehat{\Phi} S^{-1}\Phi^{T}=0.
\end{cases}
\end{equation}
Particularly, if $S=I$ (the identity matrix), then \eqref{29} shows that $L$ can be diagonalized by  using an orthogonal matrix $O(N)$; if $S=\mathrm{diag}(1,...,1,-1,...,-1)$, the diagonalization can be obtained by using a ``orthogonal'' matrix $O(p,q)$ with $p+q=N$. From the orthogonality of \eqref{210}, we obtain the eigenmatrix $\Phi$($\hat{\Phi}$) of $L$($\hat{L}$). Although eigenvalues of $L$ are real, the elements in $\Phi^{T}S\Phi$ differs from $s_{i}$.
The eigenmatrixs $\Phi$ and $\widehat{\Phi}$ consist of the eigenvectors of $L$ and $\widehat{L}$, and considering the following system of linear equations
\begin{equation}\label{211}
\begin{cases}
L\phi=\lambda\phi,\\
\widehat{L}\phi+L\widehat{\phi}=\lambda\widehat{\phi},
\end{cases}
\end{equation}
which the $\phi$ and $\widehat{\phi}$ consist of $\Phi$, $\widehat{\Phi}$ are given in the following form,
\begin{equation}\label{212}
\Phi=\left(\begin{matrix}
\phi_{1}(\lambda_{1})&\phi_{1}(\lambda_{2})&\cdots&\phi_{1}(\lambda_{N})\\
\vdots&\vdots&\vdots&\\
\phi_{N}(\lambda_{1})&\phi_{N}(\lambda_{2})&\cdots&\phi_{N}(\lambda_{N})\\
\end{matrix}
\right)\\,
\end{equation}
\begin{equation}\label{2121}
\widehat{\Phi}=\left(\begin{matrix}
\widehat{\phi}_{1}(\lambda_{1})&\widehat{\phi}_{1}(\lambda_{2})&\cdots&\widehat{\phi}_{1}(\lambda_{N})\\
\vdots&\vdots&\vdots\\
\widehat{\phi}_{N}(\lambda_{1})&\widehat{\phi}_{N}(\lambda_{2})&\cdots&\widehat{\phi}_{N}(\lambda_{N})\\
\end{matrix}
\right)\\.
\end{equation}
From the first two equations of \eqref{210}, we get something that looks like an ``orthogonality" relationship as follows,
\begin{equation}\label{213}
\begin{cases}
\sum_{k=1}^{N}s_{k}^{-1}\phi_{i}(\lambda_{k})\phi_{j}(\lambda_{k})=\delta_{ij}s_{i}^{-1},\\
\sum_{k=1}^{N}s_{k}^{-1}[\widehat{\phi}_{i}(\lambda_{k})\phi_{j}(\lambda_{k})+\phi_{i}(\lambda_{k})\widehat{\phi}_{j}(\lambda_{k})]=0.
\end{cases}
\end{equation}
Also, we can get the similarly relationship from the another two equations of \eqref{210},
\begin{equation}\label{214}
\begin{cases}
\sum_{k=1}^{N}s_{k}\phi_{k}(\lambda_{i})\phi_{k}(\lambda_{j})=\delta_{ij}s_{i},\\
\sum_{k=1}^{N}s_{k}[\widehat{\phi}_{k}(\lambda_{i})\phi_{k}(\lambda_{j})+\phi_{k}(\lambda_{i})\widehat{\phi}_{k}(\lambda_{j})]=0.
\end{cases}
\end{equation}
According to \eqref{213}, we extend the inner product with four functions of $\lambda$ {from \cite{1}},
\begin{equation}\label{215}
\begin{cases}
{<f,g>}:=\sum_{k=1}^{N}s_{k}^{-1}f(\lambda_{k})g(\lambda_{k}),\\
{<f,\widehat{g}>+<\widehat{f},g>}:=\sum_{k=1}^{N}s_{k}^{-1}[f(\lambda_{k})\widehat{g}(\lambda_{k})+\widehat{f}(\lambda_{k})g(\lambda_{k})],
\end{cases}
\end{equation}
where $\lambda$ are arbitrary.
The elements of $L$ and $\widehat{L}$ can be expressed:
\begin{equation}\label{217}
\begin{cases}
a_{ij}:=(L)_{ij}=s_{j}{<\lambda\phi_{i}\phi_{j}>},\\
\widehat{a}_{ij}:=(\widehat{L})_{ij}=s_{j}{<\lambda\widehat{\phi}_{i}\phi_{j}>+s_{j}<\lambda\phi_{i}\widehat{\phi}_{j}>}.
\end{cases}
\end{equation}
Thus, the elements of $L$ and $\widehat{L}$ can be expressed by the above inner product with $\phi_{i}$ and $\widehat{\phi}_{i}$. In fact, {a lot of work has been finished in this area. Not only the $\Phi$ can be obtained by the orthonormalization procedure of G. Szeg\"{o} \cite{11}}, but there is another way to get the orthonormalization procedure, which is introduced by Kodama and Mclauglin \cite{5}. The specific forms of $\phi(t)$ and $\widehat{\phi}(t)$ are given below,
\begin{equation}\label{218}
\phi_{i}(\lambda,t)=\frac{e^{\lambda t}}{\sqrt{D_{i}(t)D_{i-1}(t)}}\left|\begin{array}{ccccc}
s_{1}c_{11}&s_{2}c_{12}&\cdots&s_{i-1}c_{1,i-1}&\phi_{1}^{0}(\lambda)\\
\vdots&\vdots&\ddots&\vdots&\vdots\\
s_{1}c_{i1}&s_{2}c_{i2}&\cdots&s_{i-1}c_{i,i-1}&\phi_{i}^{0}(\lambda)
\end{array}\right|,
\end{equation} and
\begin{equation}\label{219}
\begin{aligned}
\widehat{\phi}_{i}(\lambda,t)&=\frac{e^{\lambda t}}{[D_{i}(t)D_{i-1}(t)]^{\frac{1}{2}}}\sum^{i-1}_{k=1}\left(\left|\begin{array}{ccccc}
s_{1}c_{11}&\cdots&s_{k}\widehat{c}_{1,k}&\cdots&s_{i-1}c_{1,i-1}\\
\vdots&\vdots&\ddots&\vdots&\vdots\\
s_{1}c_{11}&\cdots&s_{k}\widehat{c}_{1,k}&\cdots&s_{i-1}c_{i,i-1}
\end{array}\right|+\left|\begin{array}{ccccc}
s_{1}c_{11}&\cdots&\widehat{\phi}_{1}^{0}(\lambda)\\
\vdots&\vdots&\vdots\\
s_{1}c_{i1}&\cdots&\widehat{\phi}_{i}^{0}(\lambda)
\end{array}\right|\right)\\&-\frac{[\widehat{D}_{i}(t)D_{i-1}(t)+D_{i}(t)\widehat{D}_{i-1}(t)]e^{\lambda t}}{2[D_{i}(t)D_{i-1}(t)]^{\frac{3}{2}}}\left|\begin{array}{ccccc}
s_{1}c_{11}&\cdots&\phi_{1}^{0}(\lambda)\\
\vdots&\ddots&\vdots\\
s_{1}c_{i1}&\cdots&\phi_{i}^{0}(\lambda)
\end{array}\right|,
\end{aligned}
\end{equation}
where {$c_{ij}(t)=<\phi_{i}^{0}\phi_{j}^{0}e^{\lambda t}>$, $\widehat{c}_{ij}(t)=<\widehat{\phi}^{0}_{i}\phi^{0}_{j}e^{\lambda t}>+<\phi^{0}_{i}\widehat{\phi}^{0}_{j}e^{\lambda t}>$}. The $D_{k}(t)$ and $\widehat{D}_{k}(t)$ are expressed by the determinant of the $k\times k$ matrix with entries $s_{i}c_{ij}(t)$ and $s_{i}\widehat{c}_{ij}(t)$,
\begin{equation}\label{220}
D_{k}(t)=\left|\begin{array}{c}
(s_{i}c_{ij})_{1\leq i,j\leq k}
\end{array}\right|,\\
\end{equation}
\begin{equation}\label{221}
\widehat{D}_{k}(t)=\sum^{i}_{k=1}\left|\begin{array}{ccccc}
s_{1}c_{11}&\cdots&s_{k}\widehat{c}_{1,k}&\cdots&s_{i}c_{1,i}\\
\vdots&\vdots&\vdots&\vdots\\
s_{1}c_{i1}&\cdots&s_{k}\widehat{c}_{i,k}&\cdots&s_{i}c_{i,i}\\
\end{array}\right|.
\end{equation}
Note that  $D_{k}(0)=1$ and $\widehat{D}_{k}(0)=0$ for any $k$. With the {formulas} \eqref{218} and \eqref{219}, we {get} the solutions for the problem of \eqref{29},
we can derive the following proposition with the above formulas.
\begin{proposition}
$L(t)$ and $\widehat{L}(t)$ blow up to infinity at $t_{0}$ while $\widehat{D}_{i}(t_{0})=D_{i}(t_{0})=0$ with some $t_{0}$ and $i$.
 For the case of matrices $L$ and $\widehat{L}$, the $\widehat{D}_{i}(t)$ and $D_{i}(t)$ can be written with the $\tau$-functions, and the solutions $\alpha_{i}$, $\widehat{\alpha}_{i}$, $\beta_{i}$ and $\widehat{\beta}_{i}$ can be expressed in the following forms:
\end{proposition}
\begin{equation}\label{222}
\begin{cases}
\alpha_{i}=s_{i}a_{i}=\frac{d}{dt}\log\frac{\tau_{i}}{\tau_{i-1}},\\
\widehat{\alpha_{i}}=s_{i}\widehat{a}_{i}=\frac{d}{dt}\frac{\widehat{\tau}_{i}\tau_{i-1}-\tau_{i}\widehat{\tau}_{i-1}}{\tau_{i}\tau_{i-1}},
\end{cases}
\end{equation}
\begin{equation}\label{223}
\begin{cases}
\beta_{i}=s_{i}s_{i+1}b_{i}^{2}=\frac{\tau_{i+1}\tau_{i-1}}{\tau_{i}^{2}},\\
\widehat{\beta}_{i}=2s_{i}s_{i+1}b_{i}\widehat{b}_{i}=\frac{\widehat{\tau}_{i+1}\tau_{i-1}+\tau_{i+1}\widehat{\tau}_{i-1}}{\tau_{i}^{2}}-2\frac{\tau_{i+1}\tau_{i-1}\widehat{\tau}_{i}}{\tau_{i}^{3}}.
\end{cases}
\end{equation}
The derivative of the weakly coupled Toda equations \eqref{2.3} and \eqref{2.4} are expressed in the bilinear form,
\begin{equation}\label{224}
\begin{cases}
\tau_{i}\tau_{i}''-(\tau_{i}')^{2}=\tau_{i+1}\tau_{i-1},\\
\widehat{\tau}_{i}\tau''_{i}+\tau_{i}\widehat{\tau}''_{i}-2\widehat{\tau_{i}'}\tau'_{i}=\widehat{\tau}_{i+1}\tau_{i-1}+\widehat{\tau}_{i-1}\tau_{i+1}.
\end{cases}
\end{equation}
The $\tau_{i}$ and $\widehat{\tau}_{i}$ can be written as determinants refer to \cite{3,9}:
\begin{equation}\label{225}
\tau_{i}=\left|\begin{array}{ccccc}
\tau_{1}&\tau^{'}_{1}&\cdots&\tau^{(i-1)}_{1}\\
\tau^{'}_{1}&\tau^{''}_{1}&\cdots&\tau^{(i)}_{1}\\
\vdots&\vdots&\ddots&\vdots\\
\tau^{(i-1)}_{1}&\tau^{(i)}_{1}&\cdots&\tau^{(2i-2)}_{1}\\
\end{array}\right|,
\end{equation}
\begin{equation}\label{226}
\widehat{\tau}_{i}=\sum^{i-1}_{k=0}\left|\begin{array}{ccccc}
\tau_{1}&\cdots&\widehat{\tau}^{(k)}_{1}&\cdots&\tau^{(i-1)}_{1}\\
\tau^{'}_{1}&\cdots&\widehat{\tau}^{(k+1)}_{1}&\cdots&\tau^{(i)}_{1}\\
\vdots&\vdots&\vdots&\vdots&\vdots\\
\tau^{(i-1)}_{1}&\cdots&\widehat{\tau}^{(k+i-1)}_{1}&\cdots&\tau^{(2i-2)}_{1}
\end{array}\right|,
\end{equation}
where $\tau_{1}$, $\widehat{\tau}_{1}$ are given by {$\tau_{1}=c_{11}=<\phi^{0}_{1}\phi^{0}_{1}e^{\lambda t}>=s^{-1}_{1}D_{1}$, $\widehat{\tau}_{1}=\widehat{c}_{11}=<2\phi^{0}_{1}\widehat{\phi}^{0}_{1}e^{\lambda t}>=s^{-1}_{1}\widehat{D}_{1}$}. From \eqref{220} and \eqref{221}, we know that there is a certain relationship between $\tau_{i}, \widehat{\tau}_{i}$ and $D_{i}, \widehat{D}_{i}$,
\begin{equation}\label{227}
\begin{cases}
\tau_{i}=\frac{1}{s^{i}_{1}}[\prod^{i-1}_{k=1}(\beta^{0}_{k})^{i-k}]D_{i},\\
\widehat{\tau}_{i}=\frac{1}{s^{i}_{1}}\prod^{i-1}_{k=1}[(i-k)\widehat{\beta}^{0}_{k}(\beta^{0}_{k})^{i-k-1}D_{i}+(\beta^{0}_{k})^{i-k}\widehat{D}_{i}],
\end{cases}
\end{equation}
where $\beta^{0}_{k}=\beta_{k}(0)$ and $\widehat{\beta}^{0}_{k}=\widehat{\beta}_{k}(0)$. From \eqref{227}, we know that $\tau_{i}$, $\widehat{\tau}_{i}$  turn out to the blow-ups with both $\alpha_{i}$, $\widehat{\alpha}_{i}$ and $\beta_{i}$, $\widehat{\beta}_{i}$ at $t_{0}$ for $D(t_{0})=\widehat{D}(t_{0})=0$. If there are some $k$ that make $\beta^{0}_{k}=\widehat{\beta}^{0}_{k}=0$, we can not decide $\beta_{k}$, $\widehat{\beta}_{k}$ in the form of \eqref{224}, the reason is that the original system will be split into many small subsystems.
 \section{The inverse scattering method}

 \ \ \ \   Next, we are going to use the inverse scattering method to give the expressions of $\phi_{i}(\lambda,t)$ and $\widehat{\phi}_{i}(\lambda,t)$, which is mainly used in the article of Kodama \cite{1}, and the time variable of $\Phi(t)$ can be obtained by the orthonormalization procedure of Szeg\"{o} \cite{11}. {According to the reference \cite{1}}, we know that $B=L-\mathrm{diag}(L)-2(L)_{<0},\ \widehat{B}=\widehat{L}-\mathrm{diag}(\widehat{L})-2(\widehat{L})_{<0}$. Due to the first two equations of \eqref{29}:$L\Phi=\Phi\Lambda$ and $\widehat{L}\Phi+L\widehat{\Phi}=\widehat{\Phi}\Lambda$, we get $B=\Phi\Lambda\Phi^{-1}-\mathrm{diag}(L)-2(L)_{<0},$ and $\widehat{B}=\Phi\Lambda\widehat{\Phi}^{-1}+\widehat{\Phi}\Lambda\Phi^{-1}-\mathrm{diag}(\widehat{L})-2(\widehat{L})_{<0}$. Then the latter two equations of \eqref{29} can be written as:
\begin{equation}\label{31}
\begin{cases}
\frac{d}{dt}\Phi=\Phi\Lambda-[\mathrm{diag}(L)+2(L)_{<0}]\Phi,\\
\frac{d}{dt}\widehat{\Phi}=\widehat{\Phi}\Lambda-[\mathrm{diag}(\widehat{L})+2(\widehat{L})_{<0}]\Phi+[\mathrm{diag}(L)+2(L)_{<0}]\widehat{\Phi}.
\end{cases}
\end{equation}
The elements of the $\Phi_{t}$ and $\widehat{\Phi}_{t}$ are expressed by the right side of the equations \eqref{31}, and the vectors $\phi_{i}(\lambda_{k},t)$, $\widehat{\phi}_{i}(\lambda_{k},t)$ ($k=1,...,N$) of the first line in $\Phi_{t}$ and $\widehat{\Phi}_{t}$ are given by
\begin{equation}\label{32}
\begin{cases}
\frac{d}{dt}\phi_{1}(\lambda_{k})=[\lambda_{k}-s_{1}{<\lambda\phi_{1}^{2}(\lambda)>]}\phi_{1}(\lambda_{k}),\\
\frac{d}{dt}\widehat{\phi}_{1}(\lambda_{k})=[\lambda_{k}-s_{1}{<\lambda\phi^{2}_{1}(\lambda)>]\widehat{\phi}_{1}(\lambda_{k})-2s_{1}<\lambda\widehat{\phi}_{1}(\lambda)\phi_{1}(\lambda)>\phi_{1}(\lambda_{k})}.
\end{cases}
\end{equation}
Then \eqref{32} can be readily solved in the form
\begin{equation}\label{33}
\begin{cases}
\phi_{1}(\lambda_{k},t)=\frac{\psi_{1}(\lambda_{k},t)}{\sqrt{s_{1}{<\psi_{1}^{2}(\lambda,t)}>}},\\
\widehat{\phi}_{1}(\lambda_{k},t)=\frac{\widehat{\psi}_{1}(\lambda_{k},t)}
{\sqrt{s_{1}{<\psi_{1}^{2}(\lambda,t)>}}-\frac{s_{1}\psi_{1}(\lambda_{k},t)<\psi_{1}(\lambda,t)\widehat{\psi}_{1}(\lambda,t)>}
{[s_{1}<\psi_{1}^{2}(\lambda,t)>]^{\frac{3}{2}}}},
\end{cases}
\end{equation}
with $\psi_{1}(\lambda_{k},t)=\phi_{1}^{0}(\lambda_{k})e^{\lambda t}$ and $\widehat{\psi}_{1}(\lambda_{k},t)=\widehat{\phi}_{1}^{0}(\lambda_{k})e^{\lambda t}$.
The elements of the second line in $\Phi_{t}$, $\widehat{\Phi}_{t}$ are expressed,
\begin{equation}\label{34}
\begin{cases}
\frac{d}{dt}\phi_{2}(\lambda_{k})=[\lambda_{k}-s_{2}{<\lambda\phi_{2}^{2}(\lambda)>]\phi_{2}(\lambda_{k})-2s_{1}<\lambda\phi_{2}(\lambda)\phi_{1}(\lambda)>\phi_{1}(\lambda_{k})},\\
\frac{d}{dt}\widehat{\phi}_{2}(\lambda_{k})=(\lambda_{k}-s_{2}{[<\lambda\phi^{2}_{2}(\lambda)>\widehat{\phi}_{2}(\lambda_{k})-2s_{2}<\lambda\widehat{\phi}_{2}(\lambda)\phi_{2}(\lambda)>]\phi_{2}(\lambda_{k})}\\
\ \ \ \ \ \ \ \ \ \ \ \ \ \ \ -2s_{1}{<\lambda\phi_{2}(\lambda)\phi_{1}(\lambda)>\phi_{1}(\lambda_{k})-2s_{1}<\lambda\phi_{2}(\lambda)\widehat{\phi}_{1}(\lambda)+\lambda\widehat{\phi}_{2}(\lambda)\phi_{1}(\lambda)>\phi_{1}(\lambda_{k})}.
\end{cases}
\end{equation}
Through the integration, the $\phi_{2}(\lambda_{k})$, $\widehat{\phi}_{2}(\lambda_{k})$ can be written as
\begin{equation}\label{35}
\begin{cases}
\phi_{2}(\lambda_{k},t)=\frac{\psi_{2}(\lambda_{k},t)}{\sqrt{s_{2}{<\psi_{2}^{2}(\lambda,t)}>}},\\
\widehat{\phi}_{2}(\lambda_{k},t)=\frac{\widehat{\psi}_{2}(\lambda_{k},t)}{\sqrt{s_{2}{<\psi_{2}^{2}(\lambda,t)>}}-
\frac{s_{2}\psi_{2}(\lambda_{k},t)<\psi_{2}(\lambda,t)\widehat{\psi}_{2}>(\lambda,t)}{[s_{2}<\psi_{1}^{2}(\lambda,t)>]^{\frac{3}{2}}}},
\end{cases}
\end{equation}
where
\begin{equation}\label{36}
\begin{cases}
\psi_{2}(\lambda_{k},t)=\phi_{2}^{0}(\lambda_{k})e^{\lambda t}-s_{1}{<\phi_{2}^{0}(\lambda)\phi_{1}(\lambda,t)e^{\lambda t}>}\phi_{1}(\lambda_{k},t),\\
\widehat{\psi}_{2}(\lambda_{k},t)=\widehat{\phi}_{2}^{0}(\lambda_{k})e^{\lambda t}-s_{1}[{<\phi_{2}^{0}(\lambda,t)\phi_{1}(\lambda,t)e^{\lambda t}>}\widehat{\phi}_{1}(\lambda_{k},t)\\
\ \ \ \ \ \ \ \ \ \ \ \ \ \ \
+s_{1}{<\phi_{2}^{0}(\lambda)\widehat{\phi}_{1}(\lambda,t)e^{\lambda t}+\widehat{\phi}_{2}^{0}(\lambda)\phi_{1}(\lambda,t)e^{\lambda t}>}]\phi_{1}(\lambda_{k},t).
\end{cases}
\end{equation}
Generally, the $i$th lines in \eqref{31} satisfy
\begin{equation}\label{37}
\begin{cases}
\frac{d}{dt}\phi_{i}(\lambda_{k})=[\lambda_{k}-s_{i}{<\lambda\phi_{i}^{2}(\lambda)>}]\phi_{i}(\lambda_{k},t)-2\sum_{j=1}^{i-1}[s_{j}{<\lambda\phi_{i}(\lambda)\phi_{j}(\lambda)>}]\phi_{j}(\lambda_{k},t),\\
\frac{d}{dt}\widehat{\phi}_{i}(\lambda_{k})=[\lambda_{k}-s_{i}{<\lambda\phi^{2}_{i}(\lambda)>\widehat{\phi}_{i}(\lambda_{k},t)-2s_{i}<\lambda\widehat{\phi}_{i}(\lambda)\phi_{i}(\lambda)>\phi_{i}(\lambda_{k},t)}\\
\ \ \ \ \ \ \ \ \ \ \ \ \ \ \
-2\sum_{j=1}^{i-1}[{<\lambda\phi_{i}(\lambda)\phi_{j}(\lambda)>\widehat{\phi}_{j}(\lambda_{k},t)+<\lambda\widehat{\phi}_{i}(\lambda)\phi_{j}(\lambda)+\lambda\phi_{i}(\lambda)\widehat{\phi}_{j}(\lambda)>}\phi_{j}(\lambda_{k},t)].
\end{cases}
\end{equation}
Then, it is the same with the {procedures} above,
\begin{equation}\label{38}
\begin{cases}
\phi_{i}(\lambda_{k})=\frac{\psi_{i}(\lambda_{k},t)}{\sqrt{s_{i}{<\psi_{i}^{2}(\lambda,t)>}}},\\
\widehat{\phi}_{i}(\lambda_{k})=\frac{\widehat{\psi}_{i}(\lambda_{k},t)}{\sqrt{s_{i}{<\psi_{i}(\lambda,t)\psi_{j}(\lambda,t)>}}
-\frac{s_{i}\psi_{i}(\lambda_{k},t)<\widehat{\psi}_{i}(\lambda,t)\psi_{i}(\lambda,t)>}{[s_{i}<\psi^{2}_{i}(\lambda,t)>]^{\frac{3}{2}}}},
\end{cases}
\end{equation}
\begin{equation}\label{39}
\begin{cases}
\psi_{i}(\lambda_{k},t)=\phi_{i}^{0}(\lambda_{k})e^{\lambda t}-\sum_{j=1}^{i-1}s_{j}{<\phi_{i}^{0}(\lambda)\phi_{j}(\lambda,t)e^{\lambda t}>}\phi_{j}(\lambda_{k},t),\\
\widehat{\psi}_{i}(\lambda_{k},t)=\widehat{\phi}_{i}^{0}(\lambda_{k})e^{\lambda t}-\sum_{j=1}^{i-1}s_{j}[{<\phi_{i}^{0}(\lambda)\phi_{j}(\lambda,t)e^{\lambda t}>}\widehat{\psi}_{i}(\lambda_{k},t)\\
\ \ \ \ \ \ \ \ \ \ \ \ \ \
+{<\widehat{\phi}_{i}^{0}(\lambda)\phi_{j}(\lambda,t)+\phi_{i}^{0}(\lambda)\widehat{\phi}_{j}(\lambda,t)>}e^{\lambda t}]\phi_{j}(\lambda_{k},t).
\end{cases}
\end{equation}
Now, we are going to give a further proof, there are two linear equations given by
\begin{equation}\label{310}
\begin{cases}
\Phi=T\Psi,\\
\widehat{\Phi}=\widehat{T}\Psi+T\widehat{\Psi},
\end{cases}
\end{equation}
then we give
\begin{equation}\label{311}
\Psi=\left(\begin{matrix}
\psi_{1}(\lambda_{1})&\cdots&\psi_{1}(\lambda_{N})\\
\psi_{2}(\lambda_{1})&\cdots&\psi_{2}(\lambda_{N})\\
 & &\vdots\\
\psi_{N}(\lambda_{1})&\cdots&\psi_{N}(\lambda_{N})\\
\end{matrix}\right),
\end{equation}
\begin{equation}\label{311a}
\widehat{\Psi}=\left(\begin{matrix}
\widehat{\psi}_{1}(\lambda_{1})&\cdots&\widehat{\psi}_{1}(\lambda_{N})\\
\widehat{\psi}_{2}(\lambda_{1})&\cdots&\widehat{\psi}_{2}(\lambda_{N})\\
 & &\vdots\\
\widehat{\psi}_{N}(\lambda_{1})&\cdots&\widehat{\psi}_{N}(\lambda_{N})\\
\end{matrix}\right),
\end{equation}
and
\begin{equation}\label{312}
\begin{cases}
T=\mathrm{diag}[s_{i}\left<\psi_{i}^{2}\right>]^{-\frac{1}{2}},\  i=1, 2 ,\cdots, N,\\
\widehat{T}=\mathrm{diag}\big(-\frac{s_{i}\left<\widehat{\psi}_{i}\psi_{i}\right>}{[s_{1}\left<\psi^{2}_{i}\right>]^{\frac{3}{2}}}\big),\  i=1, 2 ,\cdots, N.
\end{cases}
\end{equation}
According to \eqref{310}, the expressions {of the equations \eqref{28}} can be expressed as
\begin{equation}\label{313}
\begin{cases}
(T^{-1}LT)\Psi=\Psi\Lambda,\\
(T^{-1}LT)\widehat{\Psi}+(\widehat{T}^{-1}LT+T^{-1}\widehat{L}T+T^{-1}L\widehat{T})\Psi=\widehat{\Psi}\Lambda,\\
\frac{d}{dt}\Psi=(T^{-1}BT)\Psi-(\frac{d}{dt}\mathrm{\log}T)\Psi,\\
\frac{d}{dt}\widehat{\Psi}=(T^{-1}BT)\widehat{\Psi}+(\widehat{T}^{-1}BT+T^{-1}\widehat{B}T+T^{-1}B\widehat{T})\Psi-(\frac{d}{dt}(\mathrm\log \widehat{T}\Psi+\mathrm\log T\widehat{\Psi}).
\end{cases}
\end{equation}
From \eqref{31}, we find
\begin{equation}\label{314}
\begin{cases}
T^{-1}BT=-2(T^{-1}LT)_{<0}+T^{-1}LT-(T^{-1}\mathrm{diag}(L)T),\\
\widehat{T}^{-1}BT+T^{-1}\widehat{B}T+T^{-1}B\widehat{T}=-2(\widehat{T}^{-1}LT+T^{-1}\widehat{L}T+T^{-1}L\widehat{T})_{<0}+\widehat{T}^{-1}LT\\
\ \ \ \ \ \ \ \ \ \
+T^{-1}L\widehat{T}-(\widehat{T}^{-1}\mathrm{diag}(L)T+T^{-1}\mathrm{diag}(\widehat{L})T+T^{-1}\mathrm{diag}(L)\widehat{T}).
\end{cases}
\end{equation}
Further, from \eqref{313}, we have
\begin{equation}\label{315}
\begin{cases}
\frac{d\psi}{dt}=-2(T^{-1}LT)_{<0}\psi+\lambda\psi-(\mathrm{diag}(L)+\frac{d}{dt}\log T)\psi,\\
\frac{d\widehat{\psi}}{dt}=-2[(\widehat{T}^{-1}LT)+(TL\widehat{T})]_{<0}\psi-2(T^{-1}LT)_{<0}\widehat{\psi}\\
\ \ \ \ \ \ \ \
+\lambda\widehat{\psi}-(\mathrm{diag}(L)+\frac{d}{dt}\log T)\widehat{\psi}-(\frac{d}{dt}\frac{\widehat{T}}{T})\psi,
\end{cases}
\end{equation}
where the structure of $\widehat{T}$ is similar to $T$. According to \eqref{315}, the solutions with $\psi$, $\widehat{\psi}$ expressed as
\begin{equation}\label{f}
\begin{cases}
\frac{d\psi_{i}}{dt}=-2\sum_{j=1}^{i-1}\frac{{<\lambda\psi_{i}\psi_{j}>}{<\psi_{j}^{2}>}}\psi_{j}+\lambda\psi_{i},\\
\frac{d\widehat{\psi}_{i}}{dt}=-2\sum_{j=1}^{i-1}{<\lambda\psi_{i}\psi_{j}>\frac{-2<\widehat{\psi}_{j}\psi_{j}>}{[<\psi_{j}^{2}>]^{2}}\psi_{j}
+\frac{<\lambda\psi_{i}\widehat{\psi}_{j}>+<\lambda\widehat{\psi}_{i}\psi_{j}>}{<\psi_{j}^{2}>}\psi_{j}
+\frac{<\lambda\psi_{i}\psi_{j}>}{<\psi_{j}^{2}>}}\widehat{\psi}_{j}+\lambda\widehat{\psi}_{j},
\end{cases}
\end{equation}
and
\begin{equation}\label{316}
\begin{cases}
\frac{1}{2}\frac{d}{dt}\mathrm\log{<\psi_{i}^{2}>=s_{j}<\lambda\phi_{i}^{2}>}=a_{ij},\\
\frac{1}{2}\frac{d}{dt}\mathrm\log2{<\psi_{i}\widehat{\psi}_{i}>=s_{j}[<\lambda\widehat{\phi}_{i}\phi_{i}>+<\lambda\widehat{\phi}_{j}\phi_{i}>}]=\widehat{a}_{ij}.
\end{cases}
\end{equation}
Obviously, \eqref{316} can be evolved from \eqref{f}. And through simple calculation,
\begin{equation}\label{317}
\begin{cases}
\psi(\lambda,t)=Q(t)\phi^{0}(\lambda)e^{\lambda t},\\
\widehat{\psi}(\lambda,t)=[\widehat{Q}(t)\phi^{0}(\lambda)+Q(t)\widehat{\phi}^{0}(\lambda)]e^{\lambda t},
\end{cases}
\end{equation}
where $Q(t)$, $\widehat{Q}(t)$ are  lower triangular matrices and $\phi^{0}(\lambda)=\phi(\lambda,0)$, $\widehat{\phi}^{0}(\lambda)=\widehat{\phi}(\lambda,0)$.
In fact, we have to chose the initial conditions of $\psi(\lambda,t)$, $\widehat{\psi}(\lambda,t)$ which {satisfy} $\psi(\lambda,0)=\phi^{0}(\lambda)$, $\widehat{\psi}(\lambda,0)=\widehat{\phi}^{0}(\lambda)$, and $s_{i}<\psi_{i}\psi_{j}>=s_{i}<\phi^{0}_{i}\phi^{0}_{j}>=\delta_{ij}(t=0)$, $s_{i}<\widehat{\phi}^{0}_{i}\phi^{0}_{j}+\phi^{0}_{i}\widehat{\phi}^{0}_{j}>=0(t=0)$. For the  ``orthogonality'' relations of \eqref{211} and \eqref{212}, it show that the $<\psi_{i}\psi_{j}>=0$, $<\widehat{\psi}_{i}\psi_{j}+\psi_{i}\widehat{\psi}_{j}>=0$ for $i\neq j$,  so the ``orthogonality'' relations can be written like this:
\begin{equation}\label{318}
\begin{cases}
<\psi_{i}\phi^{0}_{j}e^{\lambda t}>=\sum_{k=1}^{N}s_{k}^{-1}<\psi_{i}(\lambda_{k},t)\phi^{0}_{j}(\lambda_{k})e^{\lambda t}>=0,\\
<\widehat{\psi}_{i}\phi^{0}_{j}e^{\lambda t}>+<\psi_{i}\widehat{\phi}^{0}_{j}e^{\lambda t}>=\sum_{k=1}^{N}\left(s_{k}^{-1}<\widehat{\psi}_{i}\phi^{0}_{j}e^{\lambda t}>+s_{k}^{-1}<\psi_{i}\widehat{\phi}^{0}_{j}e^{\lambda t}>\right)=0.
\end{cases}
\end{equation}
The solutions between $\psi_{i}(\lambda,t)$ and $\widehat{\psi}_{i}(\lambda,t)$ of \eqref{316} are given from \cite{1},
\begin{equation}\label{319}
\psi_{i}(\lambda,t)=\frac{e^{\lambda t}}{D_{i-1}(t)}\left|\begin{array}{ccccc}
s_{1}c_{11}&s_{2}c_{12}&\cdots&s_{i-1}c_{1,i-1}&\phi_{1}^{0}(\lambda)\\
\vdots&\vdots&\ddots&\vdots&\vdots\\
s_{1}c_{i1}&s_{2}c_{i2}&\cdots&s_{i-1}c_{i,i-1}&\phi_{i}^{0}(\lambda)
\end{array}\right|,
\end{equation}
where $c_{ij}=\left<\phi^{0}_{i}\phi^{0}_{j}e^{\lambda t}\right>$, and the elements of $D_{k}(t)$
are $s_{j}c_{ij}(t)$, $i=1, 2,..., N$,
\begin{equation}\label{320}
D_{k}(t)=\left|(s_{i}c_{ij}(t))_{1\leq i,j\leq k}\begin{array}{ccccc}
\end{array}\right|.
\end{equation}
\begin{lemma}
According to \eqref{f}, the $\widehat{\psi}_{i}$ can be expressed
\begin{equation}\label{321}
\begin{aligned}
\widehat{\psi}_{i}(\lambda,t)&=\frac{e^{\lambda t}}{D_{i-1}(t)}\sum^{i-1}_{k=1}\left|\begin{array}{ccccc}
s_{1}c_{11}&\cdots&s_{k}\widehat{c}_{1,k}&\cdots&\phi_{1}^{0}(\lambda)\\
\vdots&\vdots&\ddots&\vdots&\vdots\\
s_{1}c_{i1}&\cdots&s_{k}\widehat{c}_{i,k}&\cdots&\phi_{i}^{0}(\lambda)
\end{array}\right|+\frac{e^{\lambda t}}{D_{i-1}(t)}\left|\begin{array}{ccccc}
s_{1}c_{11}&\cdots&\widehat{\phi}_{1}^{0}(\lambda)\\
\vdots&\vdots&\vdots\\
s_{1}c_{i1}&\cdots&\widehat{\phi}_{i}^{0}(\lambda)
\end{array}\right|\\
&-\frac{\widehat{D}_{i-1}(t)e^{\lambda t}}{D_{i-1}^{2}(t)}\left|\begin{array}{ccccc}
s_{1}c_{11}&\cdots&s_{i-1}c_{1,i-1}&\phi_{1}^{0}(\lambda)\\
\vdots&\ddots&\vdots&\vdots\\
s_{1}c_{i1}&\cdots&s_{i-1}c_{i,i-1}&\phi_{i}^{0}(\lambda)
\end{array}\right|,
\end{aligned}
\end{equation}
where $\widehat{c}_{ij}(t)={<\widehat{\phi}^{0}_{i}\phi^{0}_{j}e^{\lambda t}>+<\phi^{0}_{i}\widehat{\phi}^{0}_{j}e^{\lambda t}>}$, and $k$ represents the number of columns in the determinant above,
\begin{equation}\label{322}
\widehat{D}_{k}(t)=\sum^{i}_{k=1}\left|\begin{array}{ccccc}
s_{1}c_{11}&\cdots&s_{k}\widehat{c}_{1,k}&\cdots&s_{i}c_{1,i}\\
\vdots&\vdots&\vdots&\vdots\\
s_{1}c_{i1}&\cdots&s_{k}\widehat{c}_{i,k}&\cdots&s_{i}c_{i,i}\\
\end{array}\right|.
\end{equation}
\end{lemma}
\begin{proof}
From \eqref{318} and \eqref{317}, we have
\begin{equation}\label{323}
\begin{cases}
s_{j}\sum_{k=1}^{i}Q_{ik}c_{kj}(t)=0,\\
s_{j}\sum_{k=1}^{i}[\widehat{Q}_{ik}c_{kj}(t)+Q_{ik}\widehat{c}_{kj}(t)]=0, \hspace{1cm} 1\leq j\leq i-1.
\end{cases}
\end{equation}
Solving \eqref{323} for $Q_{ik}$ and $\widehat{Q}_{ik}$, we have
\begin{equation}\label{324}
\begin{cases}
Q_{ik}=-\frac{D^{k}_{i-1}(t)}{D_{i-1}(t)},\\
\widehat{Q}_{ik}=-\frac{\widehat{D}^{k}_{i-1}(t)}{D_{i-1}(t)}+\frac{\widehat{D}_{i-1}(t)D^{k}_{i-1}(t)}{D^{2}_{i-1}(t)}.
\end{cases}
\end{equation}
In fact, $\widehat{D}^{k}_{i-1}(t)(D^{k}_{i-1}(t)$) is the substitution of the $k$th and the $i$th row of $\widehat{D}_{i-1}(t)(D_{i-1}(t))$. From \eqref{317}, we have
\begin{align}\label{325}
\widehat{\psi}_{i}&=e^{\lambda t}\sum_{k=1}^{i}(Q_{ik}\widehat{\phi}^{0}_{k}+\widehat{Q}_{ik}\phi^{0}_{k})\\
&=-e^{\lambda t}\sum_{k=1}^{i-1}\left(\frac{\widehat{D}_{i-1}(t)D^{k}_{i-1}(t)-D_{i-1}(t)\widehat{D}^{k}_{i-1}(t)}{D^{2}_{i-1}(t)}\phi^{0}_{k}
+e^{\lambda t}\frac{D^{k}_{i-1}(t)}{D_{i-1}(t)}\widehat{\phi}^{0}_{k}\right)+e^{\lambda t}\frac{D_{i-1}(t)}{D_{i-1}(t)}\widehat{\phi}^{0}_{i}\\
&=\frac{e^{\lambda t}}{D_{i-1}}\left(\sum^{i-1}_{k=1}(-1)^{i+k}\widehat{\phi}^{0}_{k}\left|\begin{array}{ccccc}
s_{1}c_{11}&\cdots&s_{1}c_{1i}&\cdots&s_{1}c_{1,i-1}\\
\vdots&\cdots&\vdots&\cdots&\vdots\\
s_{i-1}c_{i-1,1}&\cdots&s_{i-1}c_{i-1,i}&\cdots&s_{i-1}c_{i-1,i-1}
\end{array}\right|+\widehat{\phi}^{0}_{i}D_{i-1}(t)\right)\\
&-e^{\lambda t}\sum_{k=1}^{i-1}\frac{\widehat{D}_{i-1}(t)D^{k}_{i-1}(t)-D_{i-1}(t)\widehat{D}^{k}_{i-1}(t)}{D^{2}_{i-1}(t)}\phi^{0}_{k},
\end{align}

which is just \eqref{321}. From \eqref{319} and \eqref{321}, $<\psi^{2}_{i}>$, $<2\widehat{\psi}_{i}\psi_{i}>$ can be expressed with $D_{i}$ and $\widehat{D}_{i}$
\begin{equation}\label{326}
\begin{cases}
<\psi^{2}_{i}>=\frac{D_{i}}{s_{i}D_{i-1}},\\
<2\widehat{\psi}_{i}\psi_{i}>=\frac{D_{i-1}\widehat{D}_{i}-\widehat{D}_{i-1}D_{i}}{s_{i}D^{2}_{i-1}}.
\end{cases}
\end{equation}
Then we can obtain the formulas that we have mentioned above,
\begin{equation}\label{327}
\phi_{i}(\lambda,t)=\frac{e^{\lambda t}}{\sqrt{D_{i}(t)D_{i-1}}(t)}\left|\begin{array}{ccccc}
s_{1}c_{11}&s_{2}c_{12}&\cdots&s_{i-1}c_{1,i-1}&\phi_{1}^{0}(\lambda)\\
\vdots&\vdots&\ddots&\vdots&\vdots\\
s_{1}c_{i1}&s_{2}c_{i2}&\cdots&s_{i-1}c_{i,i-1}&\phi_{i}^{0}(\lambda)
\end{array}\right|,
\end{equation}
\begin{equation}\label{328}
\begin{aligned}
\widehat{\phi}_{i}(\lambda,t)&=\frac{e^{\lambda t}}{[D_{i}(t)D_{i-1}(t)]^{\frac{1}{2}}}\left(\sum^{i-1}_{k=1}\begin{vmatrix}
s_{1}c_{1,i-1}&\cdots&s_{k}\widehat{c}_{1,k}&\cdots&s_{i-1}c_{1,i-1}\\
\vdots&\vdots&\ddots&\ddots&\vdots\\
s_{1}c_{1,i-1}&\cdots&s_{k}\widehat{c}_{i,k}&\cdots&s_{i-1}c_{i,i-1}
\end{vmatrix}+\left|\begin{array}{ccccc}
s_{1}c_{11}&\cdots&\widehat{\phi}_{1}^{0}(\lambda)\\
\vdots&\vdots&\vdots\\
s_{1}c_{i1}&\cdots&\widehat{\phi}_{i}^{0}(\lambda)
\end{array}\right|\right)\\&-\frac{[\widehat{D}_{i}(t)D_{i-1}(t)+D_{i}(t)\widehat{D}_{i-1}(t)]e^{\lambda t}}{2[D_{i}(t)D_{i-1}(t)]^{\frac{3}{2}}}\left|\begin{array}{ccccc}
s_{1}c_{11}&\cdots&\phi_{1}^{0}(\lambda)\\
\vdots&\ddots&\vdots\\
s_{1}c_{i1}&\cdots&\phi_{i}^{0}(\lambda)
\end{array}\right|.
\end{aligned}
\end{equation}
\end{proof}
By using the formulas \eqref{328} and \eqref{327},  we {get the solutions} for \eqref{29}. The method above is similar to szeg\"{o} \cite{11}, and it is also equivalent to the procedure of Gram-Schmidt \cite{7}.
\subsection{Example}
 \ \ \ \ Next, we give a simple example to verify our results, and talk about the properties of the solutions. Let $L_{2}$, $\widehat{L}_{2}$ are a $2\times2$ matrix with $S=\mathrm{diag}(1,-1)$, and {$L_{2}$, $\widehat{L}_{2}$} are given by
\begin{equation}\label{41}
L_{2}=\left(\begin{matrix}
a_{1}&-b_{1}\\
b_{1}&-a_{2}\\
\end{matrix}
\right)\\,
\end{equation}
\begin{equation}\label{41a}
\widehat{L}_{2}=\left(\begin{matrix}
\widehat{a}_{1}&-\widehat{b}_{1}\\
\widehat{b}_{1}&-\widehat{a}_{2}\\
\end{matrix}
\right)\\.
\end{equation}
According to \eqref{2.3} and \eqref{2.4}, we have
\begin{equation}\label{42}
\begin{cases}
\frac{da_{1}}{dt}=-\frac{1}{2}b^{2}_{1}-\frac{1}{2}\widehat{b}^{2}_{1},\\
\frac{d\widehat{a}_{1}}{dt}=-b_{1}\widehat{b}_{1},\\
\frac{db_{1}}{dt}=-\frac{1}{4}b_{1}(a_{1}+a_{2})-\frac{1}{4}\widehat{b}_{1}(\widehat{a}_{1}+\widehat{a}_{2}),\\
\frac{d\widehat{b}_{1}}{dt}=-\frac{1}{4}[(a_{1}+a_{2})\widehat{b}_{1}+(\widehat{a}_{1}+\widehat{a}_{2})b_{1}],\\
\frac{da_{2}}{dt}=-\frac{1}{2}b^{2}_{1},\\
\frac{d\widehat{a}_{2}}{dt}=-b_{1}\widehat{b}_{1}-\frac{1}{2}\widehat{b}^{2}_{1},\\
\end{cases}
\end{equation}
note that, $b_{2}=\widehat{b}_{2}=0$. From \eqref{42}, if the initial value of ($a_{1}+a_{2}$) is positive, one can find $(a_{1}+a_{2})\longrightarrow\infty$ , then $b_{1}$, $\widehat{b}_{1}$ increase bigger and faster, and the corresponding result is that $\widehat{a}_{1}\longrightarrow\infty$ . On the contrary, if the initial value of $a_{1}+a_{2}\longrightarrow-\infty$, then $b_{1}$ and $\widehat{b}_{1}\longrightarrow 0$ .

The following we give some special values to $a_{i}$, $\widehat{a}_{i}$, $b_{i}$ and $\widehat{b}_{i}$ for (i=1, 2), and keep one parameter $m$ in $L$, then we discuss its eigenvalues and eigenvectors of $L$ and their properties at the same time. While $a_{1}$, $\widehat{a}_{1}$ and $\widehat{b}_{1}$ are equal to $0$, take $b_{1}$ and $\widehat{a}_{2}$  are equal to $1$, and take $a_{2}$ as a parameter, then we have
\begin{equation}\label{43}
L_{2}=\left(\begin{matrix}
0&-1\\
1&-m\\
\end{matrix}
\right)\\,
\end{equation}
\begin{equation}\label{43a}
\widehat{L}_{2}=\left(\begin{matrix}
0&0\\
0&-1 \\
\end{matrix}
\right)\\.
\end{equation}
And what is  more, we can get the eigenvalues and eigenvectors of the particular matrix, which the specific forms are given as follow,
\begin{equation}\label{44}
\begin{cases}
\lambda_{1}=\frac{1}{2}(\sqrt{m^{2}-4}-m),\\
\lambda_{2}=\frac{1}{2}(-\sqrt{m^{2}-4}-m),\\
\widehat{\lambda}_{1}=0,\\
\widehat{\lambda}_{2}=-1,
\end{cases}
\end{equation}
\begin{equation}\label{45}
\Phi^{0}_{2}=\left(\begin{matrix}
\lambda_{1}&\lambda_{2}\\
-1&-1\\
\end{matrix}
\right)\\,
\end{equation}
\begin{equation}\label{45a}
\widehat{\Phi}^{0}_{2}=\left(\begin{matrix}
1&0\\
0&1\\
\end{matrix}
\right)\\,
\end{equation}
where $\lambda_{i}$ and $\widehat{\lambda}_{i} (i=1, 2)$ are characteristic value. Then we discuss the value of $m$ below, different values of $m$ produce different results: (1) when $m\geq2$, then $0\geq\lambda_{1}\geq\lambda_{2}$; (2) when $\mid m\mid<2$, then $\lambda_{1}$ and $\lambda_{2}$ are complex; (3) when $m\leq-2$, then $\lambda_{2}>\lambda_{1}>0$, $\lambda_{2}$ and $\lambda_{1}$ are real, that is exactly the cases we need. According to \eqref{44}, \eqref{45} and \eqref{45a}, we can get some specific results:
\begin{equation}\label{46}
\Phi_{2}(t)=\frac{1}{e^{(\lambda_{2}+2\lambda_{1})t}}\left(\begin{matrix}
\lambda_{1}e^{\lambda_{1}t}&\lambda_{2}e^{\lambda_{2}t}\\
-e^{\lambda_{1}t}&-e^{\lambda_{2}t}\\
\end{matrix}
\right)\\,
\end{equation}
\begin{equation}\label{46a}
\widehat{\Phi}_{2}(t)=\left(\begin{matrix}
\frac{2\lambda_{1}+1}{e^{\frac{1}{2}\lambda_{1}t}}+\frac{2\lambda_{1}+e^{\lambda_{1}t}}{(2e^{\lambda_{1}t})^{\frac{3}{2}}}&(\frac{2\lambda_{1}+1}{e^{\frac{1}{2}\lambda_{1}t}}+\frac{2\lambda_{1}+e^{\lambda_{1}t}}{(2e^{\lambda_{1}t})^{\frac{3}{2}}})e^{t}\\
2\frac{(\lambda^{2}_{1}+1)e^{(\lambda_{1}+1)}-(\lambda_{2}+\lambda_{1})e^{\lambda_{2}t}}{[(\lambda_{1}+\lambda_{2}+2)e^{(2\lambda_{1}+\lambda_{2})t}]^{\frac{1}{2}}}&2\frac{(\lambda^{2}_{1}+1)e^{(\lambda_{1}+1)}-(\lambda_{2}+\lambda_{1})e^{\lambda_{2}t}}{[(\lambda_{1}+\lambda_{2}+2)e^{(2\lambda_{1}+\lambda_{2})t}]^{\frac{1}{2}}}e^{t}\\
\end{matrix}
\right)=\left(\begin{matrix}
\widehat{\Phi}_{1,1}&\widehat{\Phi}_{1,2}\\
\widehat{\Phi}_{2,1}&\widehat{\Phi}_{2,2}\\
\end{matrix}
\right)\\,
\end{equation}
where $\widehat{\Phi}_{i,j}(t) (1\leq i,j\leq2)$ are obtained from \eqref{46a}. The solutions of the extended Toda equations are obtained from \eqref{218} and \eqref{219},
\begin{equation}\label{47}
L_{2}(t)=\left(\begin{matrix}
\lambda_{2}e^{2(\lambda_{1}+\lambda_{2})t}+\lambda_{1}e^{4\lambda_{1}t}&-e^{2(\lambda_{1}+\lambda_{2})t}-e^{4\lambda_{1}t}\\
-e^{2(\lambda_{1}+\lambda_{2})t}-e^{4\lambda_{1}t}&\lambda_{1}e^{2(\lambda_{1}+\lambda_{2})t}+\lambda_{2}e^{4\lambda_{1}t}\\
\end{matrix}
\right)\\,
\end{equation}
\begin{equation}\label{47a}
\widehat{L}_{2}(t)=\frac{1}{e^{(\lambda_{2}+2\lambda_{1})t}}\left(\begin{matrix}
\widehat{a}_{1,1}(t)&\widehat{a}_{1,2}(t)\\
\widehat{a}_{2,1}(t)&\widehat{a}_{2,2}(t)\\
\end{matrix}
\right)\\,
\end{equation}
where
\begin{equation}
\begin{aligned}
\widehat{a}_{1,1}(t)=[\lambda_{1}e^{(3\lambda_{1}+2\lambda_{2})t}+\lambda_{1}e^{4\lambda_{1}(t)}]\widehat{\Phi}_{2,1}(t)-[e^{(3\lambda_{1}+2\lambda_{2})t}+\lambda^{2}_{1}e^{5\lambda_{1}(t)}]\widehat{\Phi}_{1,1}(t)&\\+[\lambda_{2}e^{(3\lambda_{1}+2\lambda_{2})t}+\lambda_{1}e^{5\lambda_{1}(t)}-e^{\lambda_{1}t}]\widehat{\Phi}_{1,2}(t)
-[e^{(3\lambda_{1}+2\lambda_{2})t}+\lambda^{2}_{1}e^{5\lambda_{1}(t)}]\widehat{\Phi}_{2,2}(t),
\end{aligned}
\end{equation}
\begin{equation}
\begin{aligned}
\widehat{a}_{1,2}(t)=[\lambda_{2}e^{(3\lambda_{2}+2\lambda_{1})t}+\lambda_{2}e^{(4\lambda_{1}+\lambda_{2})(t)}]\widehat{\Phi}_{2,1}(t)-[\lambda^{2}_{2}e^{(3\lambda_{2}+2\lambda_{1})t}+e^{(4\lambda_{1}+\lambda_{2})t}]\widehat{\Phi}_{1,1}(t)&\\+[\lambda_{2}e^{(3\lambda_{2}+2\lambda_{1})t}+\lambda_{2}e^{(4\lambda_{1}+\lambda_{2})t}]\widehat{\Phi}_{2,2}(t)
-[\lambda^{2}_{2}e^{(3\lambda_{2}+\lambda_{1})t}+e^{(4\lambda_{1}+\lambda_{2})t}-\lambda_{2}e^{\lambda_{2}t}]\widehat{\Phi}_{1,2}(t),
\end{aligned}
\end{equation}
\begin{equation}
\begin{aligned}
\widehat{a}_{2,1}(t)=[\lambda_{1}e^{(3\lambda_{1}+2\lambda_{2})t}+\lambda_{1}e^{5\lambda_{1}(t)}]\widehat{\Phi}_{1,1}(t)-[\lambda^{2}_{1}e^{(3\lambda_{1}+2\lambda_{2})t}+e^{5\lambda_{1}(t)}]\widehat{\Phi}_{2,1}(t)&\\+[\lambda_{1}e^{(3\lambda_{1}+2\lambda_{2})t}+\lambda_{2}e^{5\lambda_{1}(t)}-e^{\lambda_{1}t}]\widehat{\Phi}_{2,2}(t)-[e^{(3\lambda_{1}+2\lambda_{2})t}+e^{5\lambda_{1}(t)}]\widehat{\Phi}_{1,2}(t),
\end{aligned}
\end{equation}
\begin{equation}
\begin{aligned}
\widehat{a}_{2,2}(t)=[\lambda_{2}e^{(3\lambda_{2}+2\lambda_{1})t}+\lambda_{2}e^{(4\lambda_{1}+\lambda_{2})(t)}]\widehat{\Phi}_{1,1}(t)-[e^{(3\lambda_{2}+2\lambda_{1})t}+\lambda^{2}_{2}e^{(4\lambda_{1}+\lambda_{2})t}]\widehat{\Phi}_{2,1}(t)&\\+[\lambda_{1}e^{(3\lambda_{2}+2\lambda_{1})t}+\lambda_{2}e^{(4\lambda_{1}+\lambda_{2})t}-e^{\lambda_{2}t}]\widehat{\Phi}_{2,2}(t)
-[e^{(3\lambda_{2}+2\lambda_{1})t}+e^{(4\lambda_{1}+\lambda_{2})t}]\widehat{\Phi}_{1,2}(t).
\end{aligned}
\end{equation}
In addition,  according to \eqref{46a}, there is a situation that is $m=\pm2$, one of the things to pay attention is that $L_{2}(t)\nrightarrow \pm$ $\mathrm{diag}(1, 1)$.
\section{Strongly coupled Toda lattices with indefinite metrics}

 \ \ \ \   In this section, we introduce a new strongly coupled Toda lattices with indefinite metrics. For the Hamiltonian \eqref{a2}, we give a extended transformation of variables,
\begin{equation}\label{6.1}
\begin{cases}
s_{k}a_{k}=-\frac{y_{k}}{2},\\
s_{k}\widetilde{a}_{k}=-\frac{\widetilde{y}_{k}}{2},k=1,\ldots, N,
\end{cases}
\end{equation}
\begin{equation}\label{6.2}
\begin{cases}
b_{k}=\frac{1}{2}\exp(\frac{x_{k}-x_{k+1}}{2})\cosh(\frac{\widetilde{x}_{k}-\widetilde{x}_{k+1}}{2}),\\
\widetilde{b}_{k}=\frac{1}{2}\exp(\frac{x_{k}-x_{k+1}}{2})\sinh(\frac{\widetilde{x}_{k}-\widetilde{x}_{k+1}}{2}), k=1,\ldots,N-1.
\end{cases}
\end{equation}
In addition, the strongly coupled Toda lattices with indefinite metrics can be expressed as:
\begin{equation}\label{6.3}
\begin{cases}
\frac{da_{k}}{dt}=[s_{k+1}(b_{k}^{2}+\widetilde{b}_{k}^{2})-s_{k-1}(b_{k-1}^{2}+\widetilde{b}_{k-1}^{2})],\\
\frac{d\widetilde{a}_{k}}{dt}=2s_{k+1}b_{k}\widetilde{b}_{k}+2s_{k-1}b_{k-1}\widetilde{b}_{k-1},
\end{cases}
\end{equation}
\begin{equation}\label{6.4}
\begin{cases}
\frac{db_{k}}{dt}=\frac{1}{2}[s_{k+1}(b_{k}a_{k+1}+\widetilde{b}_{k}\widetilde{a}_{k+1})-s_{k}(b_{k}a_{k}+\widetilde{b}_{k}\widetilde{a}_{k})],\\
\frac{d\widetilde{b}_{k}}{dt}=\frac{1}{2}[s_{k+1}(\widetilde{b}_{k}a_{k+1}+b_{k}\widetilde{a}_{k+1})-s_{k}(\widetilde{b}_{k}a_{k}+b_{k}\widetilde{a}_{k})],
\end{cases}
\end{equation}
where $b_{0}=\widetilde{b}_{0}=b_{N}=\widetilde{b}_{N}=0$. According to the strongly coupled Toda lattices with indefinite metrics above, we can use Lax pair to express it as following,
\begin{equation}\label{6.5}
\begin{cases}
\frac{d}{dt}L=[B,L]+[\widetilde{B},\widetilde{L}],\\
\frac{d}{dt}\widetilde{L}=[\widetilde{B},L]+[B,\widetilde{L}],
\end{cases}
\end{equation}
where $L$ , $\widetilde{L}$  have the following form:
\begin{equation}\label{6.6}
L=\left(\begin{smallmatrix}
s_{1}a_{1}&0&s_{2}b_{1}&0&\cdots&0&0\\
0&s_{1}a_{1}&0&s_{2}b_{1}&\cdots&0&0\\
s_{1}b_{1}&0&s_{2}a_{2}&0&\cdots&0&0\\
0&s_{1}b_{1}&0&s_{2}a_{2}&\cdots&0&0\\
 & &\ddots&\ddots&\ddots&\\
0&\cdots&\cdots&0&0&s_{N}a_{N}&0\\
0&\cdots&\cdots&0&s_{N}b_{N-1}&0&s_{N}a_{N}\\
\end{smallmatrix}
\right),
\end{equation}
\begin{equation}\label{6.7}
\widetilde{L}=\left(\begin{smallmatrix}
0&s_{1}\widetilde{a}_{1}&0&s_{2}\widetilde{b}_{1}&\cdots&0&0\\
s_{1}\widetilde{a}_{1}&0&s_{2}\widetilde{b}_{1}&0&\cdots&0&0\\
0&s_{1}\widetilde{b}_{1}&0&s_{2}\widetilde{a}_{2}&\cdots&0&0\\
s_{1}\widetilde{b}_{1}&0&s_{2}\widetilde{a}_{2}&0&\cdots&0&0\\
 & &\ddots&\ddots&\ddots&\\
0&\cdots&\cdots&0&s_{N-1}\widetilde{b}_{N-1}&0&s_{N}\widetilde{a}_{N}\\
0&\cdots&\cdots&0&0&s_{N}\widetilde{a}_{N}&0\\
\end{smallmatrix}
\right),
\end{equation}
and $B$, $\widetilde{B}$ are given by
\begin{equation}\label{6.8}
\begin{cases}
B:=\frac{1}{2}[(L)_{>0}-(L)_{<0}],\\
\widetilde{B}:=\frac{1}{2}[(\widetilde{L})_{>0}-(\widetilde{L})_{<0}].
\end{cases}
\end{equation}
For Lax equations \eqref{6.5}, we get some linear equations,
\begin{equation}\label{6.9}
\begin{cases}
L\Phi+\widetilde{L}\widetilde{\Phi}=\Phi\Lambda,\\
\widetilde{L}\Phi+L\widetilde{\Phi}=\widetilde{\Phi}\Lambda,\\
\frac{d}{dt}\Phi=B\Phi+\widetilde{B}\widetilde{\Phi},\\
\frac{d}{dt}\widetilde{\Phi}=\widetilde{B}\Phi+B\widetilde{\Phi},
\end{cases}
\end{equation}
where $\Phi$ is the eigenmatrix of $L$, $\widetilde{\Phi}$ is the eigenmatrix of $\widetilde{L}$, and $\Lambda$ is a diagonal matrix. Thus, $\widetilde{\Phi}$ and $\Phi$ satisfy some relationships:
\begin{equation}\label{6.10}
\begin{cases}
\Phi S^{-1}\Phi^{T}+\widetilde{\Phi}S^{-1}\widetilde{\Phi}^{T}=S^{-1},\\
\widetilde{\Phi}S^{-1}\Phi^{T}+\Phi S^{-1}\widetilde{\Phi}^{T}=0,\\
\Phi^{T}S\Phi+\widetilde{\Phi}^{T}S\widetilde{\Phi}=S,\\
\widetilde{\Phi}^{T}S\Phi+\Phi^{T}S\widetilde{\Phi}=0,
\end{cases}
\end{equation}
where \begin{equation}\label{6.11}
\begin{cases}
\Phi=[\phi_{i}(\lambda_{j})]_{1\leq i,j\leq N},\\
\widetilde{\Phi}=[\widetilde{\phi}_{i}(\lambda_{j})]_{1\leq i,j\leq N}.
\end{cases}
\end{equation}
Form the equations of \eqref{6.10}, one can get the relationships,
\begin{equation}\label{6.12}
\begin{cases}
\sum_{k=1}^{N}s_{k}^{-1}[\phi_{i}(\lambda_{k})\phi_{j}(\lambda_{k})+\widetilde{\phi}_{i}(\lambda_{k})\widetilde{\phi}_{j}(\lambda_{k})]=\delta_{ij}s_{i}^{-1},\\
\sum_{k=1}^{N}s_{k}^{-1}[\widetilde{\phi}_{i}(\lambda_{k})\phi_{j}(\lambda_{k})+\phi_{i}(\lambda_{k})\widetilde{\phi}_{j}(\lambda_{k})]=0,\\
\sum_{k=1}^{N}s_{k}[\phi_{k}(\lambda_{i})\phi_{k}(\lambda_{j})+\widetilde{\phi}_{k}(\lambda_{i})\widetilde{\phi}_{k}(\lambda_{j})]=\delta_{ij}s_{i},\\
\sum_{k=1}^{N}s_{k}[\widetilde{\phi}_{k}(\lambda_{i})\phi_{k}(\lambda_{j})+\phi_{k}(\lambda_{i})\widetilde{\phi}_{k}(\lambda_{j})]=0.
\end{cases}
\end{equation}
So, the extended matrices of $L$ and $\widetilde{L}$ are expressed  by
\begin{equation}\label{6.13}
\begin{cases}
a_{ij}:=(L)_{ij}=s_{j}{<\lambda\phi_{i}\phi_{j}+\lambda\widetilde{\phi}_{i}\widetilde{\phi}_{j}>},\\
\widetilde{a}_{ij}:=(\widetilde{L})_{ij}=s_{j}{<\lambda\widetilde{\phi}_{i}\phi_{j}+\lambda\phi_{i}\widetilde{\phi}_{j}>}.
\end{cases}
\end{equation}
From the inverse scattering method, two new explicit forms of $\Phi(t)$, $\widetilde{\Phi}(t)$ are given as following
\begin{equation}\label{6.14}
\begin{aligned}
\phi_{i}(\lambda,t)&=M\left|\begin{array}{ccccc}
s_{1}c_{11}&\cdots&s_{i-1}c_{1,i-1}&\phi^{0}_{1}\\
\vdots&\vdots&\vdots&\vdots\\
s_{1}c_{i1}&\cdots&s_{i-1}c_{i,i-1}&\phi^{0}_{i}\\
\end{array}\right|+M\sum^{[\frac{i}{2}]}_{q=0}\sum_{\sum\limits_{j=1}^{i-1}k_{j}=2q}\left|\begin{array}{ccccc}
s_{1}c^{[k_{1}]}_{11}&\cdots&s_{i-1}c^{[k_{i-1}]}_{1,i-1}&\phi^{0}_{1}\\
\vdots&\vdots&\vdots&\vdots\\
s_{1}c^{[k_{1}]}_{i1}&\cdots&s_{i-1}c^{[k_{i-1}]}_{i,i-1}&\phi^{0}_{i}\\
\end{array}\right|
\\&\widehat{M}
\sum^{[\frac{i}{2}]}_{q=0}\sum_{\sum\limits_{j=1}^{i-1}k_{j}=2q-1}\left|\begin{array}{ccccc}
s_{1}c^{[k_{1}]}_{11}&\cdots&s_{i-1}c^{[k_{i-1}]}_{1,i-1}&\phi^{0}_{1}\\
\vdots&\vdots&\vdots&\vdots\\
s_{1}c^{[k_{1}]}_{i1}&\cdots&s_{i-1}c^{[k_{i-1}]}_{i,i-1}&\phi^{0}_{i}\\
\end{array}\right|,
\end{aligned}
\end{equation}
where $M=\frac{H_{0}e^{\lambda t}}{H^{2}_{0}-H^{2}_{1}}, \widehat{M}=-\frac{H_{1}e^{\lambda t}}{H^{2}_{0}-H^{2}_{1}}$,
and
\begin{equation}\label{6.15}
\begin{aligned}
\widetilde{\phi}_{i}(\lambda,t)&=\widehat{M}\left|\begin{array}{ccccc}
s_{1}c_{11}&\cdots&s_{i-1}c_{1,i-1}&\phi^{0}_{1}\\
\vdots&\vdots&\vdots&\vdots\\
s_{1}c_{i1}&\cdots&s_{i-1}c_{i,i-1}&\phi^{0}_{i}\\
\end{array}\right|+\widehat{M}\sum^{[\frac{i}{2}]}_{q=0}\sum_{\sum\limits_{j=1}^{i-1}k_{j}=2q}\left|\begin{array}{ccccc}
s_{1}c^{[k_{1}]}_{11}&\cdots&s_{i-1}c^{[k_{i-1}]}_{1,i-1}&\phi^{0}_{1}\\
\vdots&\vdots&\vdots&\vdots\\
s_{1}c^{[k_{1}]}_{i1}&\cdots&s_{i-1}c^{[k_{i-1}]}_{i,i-1}&\phi^{0}_{i}\\
\end{array}\right|
\\&+M
\sum^{[\frac{i}{2}]}_{q=0}\sum_{\sum\limits_{j=1}^{i-1}k_{j}=2q-1}\left|\begin{array}{ccccc}
s_{1}c^{[k_{1}]}_{11}&\cdots&s_{i-1}c^{[k_{i-1}]}_{1,i-1}&\phi^{0}_{1}\\
\vdots&\vdots&\vdots&\vdots\\
s_{1}c^{[k_{1}]}_{i1}&\cdots&s_{i-1}c^{[k_{i-1}]}_{i,i-1}&\phi^{0}_{i}\\
\end{array}\right|,
\end{aligned}
\end{equation}
where $H_{0}=\frac{\sqrt{2}b}{\sqrt{a+(a^{2}-b^{2})^{\frac{1}{2}}}}$, $H_{1}=\frac{\sqrt{a+(a^{2}-b^{2})^{\frac{1}{2}}}}{\sqrt{2}}$, and $a=D_{i}(t)D_{i-1}(t)+\widetilde{D}_{i}(t)\widetilde{D}_{i-1}(t)$, $b=\widetilde{D}_{i}(t)D_{i-1}(t)+D_{i}(t)\widetilde{D}_{i-1}(t)$, and
\begin{equation}\label{6.16}
\begin{aligned}
\widetilde{D}_{k}(t)&=\sum^{[\frac{i}{2}]}_{q=0}\sum_{\sum\limits_{j=1}^{i}k_{j}=2q-1}\left|\begin{array}{ccccc}
s_{1}c^{[k_{1}]}_{11}&\cdots&s_{p}c^{[k_{p}]}_{1,p}&\cdots&s_{i}c^{[k_{i}]}_{1,i}\\
\vdots&\vdots&\vdots&\vdots\\
s_{1}c^{[k_{1}]}_{i1}&\cdots&s_{p}c^{[k_{p}]}_{i,p}&\cdots&s_{i}c^{[k_{i}]}_{i,i}\\
\end{array}\right|,
\end{aligned}
\end{equation}
\begin{equation}\label{6.17}
\begin{aligned}
D_{k}(t)&=\sum^{[\frac{i}{2}]}_{q=0}\sum_{\sum\limits_{j=1}^{i}k_{j}=2q}\left|\begin{array}{ccccc}
s_{1}c^{[k_{1}]}_{11}&\cdots&s_{p}c^{[k_{p}]}_{1,p}&\cdots&s_{i}c^{[k_{i}]}_{1,i}\\
\vdots&\vdots&\vdots&\vdots\\
s_{1}c^{[k_{1}]}_{i1}&\cdots&s_{p}c^{[k_{p}]}_{i,p}&\cdots&s_{i}c^{[k_{i}]}_{i,i}\\
\end{array}\right|,
\end{aligned}
\end{equation}
and $c^{[k_{p}]}_{ij}=
\begin{cases}
c_{ij}, k_{p}=0\\
\widetilde{c}_{ij}, k_{p}=1.
\end{cases}
$ Thus we obtain the solutions \eqref{6.14} and \eqref{6.15} for the problem \eqref{6.9}. In fact, for the matrices $L$, $\widetilde{L}$, the determinants $D_{i}(t)$, $\widetilde{D}_{i}(t)$ can be written with $\tau$-functions, then $\alpha_{i}$, $\widetilde{\alpha}_{i}$, $\beta_{i}$ and $\widetilde{\beta}_{i}$ are expressed as
\begin{equation}\label{6.18}
\begin{cases}
\alpha_{i}=s_{i}a_{i}=\frac{d}{dt}\log\frac{(\tau_{i}-\widetilde{\tau}_{i})(\tau_{i-1}-\widetilde{\tau}_{i-1})}{\tau^{2}_{i-1}-\widetilde{\tau}^{2}_{i-1}},\\
\widetilde{\alpha}_{i}=s_{i}\widetilde{a}_{i}=\frac{1}{2}\frac{d}{dt}\log\frac{(\tau_{i}+\widetilde{\tau}_{i})(\tau_{i-1}-\widetilde{\tau}_{i-1})}{(\tau_{i}-\widetilde{\tau}_{i})(\tau_{i-1}+\widetilde{\tau}_{i-1})},\\
\beta_{i}=s_{i}s_{i+1}(b_{i}^{2}+\widetilde{b}_{i}^{2})=\frac{(\tau_{i+1}\tau_{i-1}+\widetilde{\tau}_{i+1}\widetilde{\tau}_{i-1})(\tau^{2}_{i}+\widetilde{\tau}^{2}_{i})-2(\widetilde{\tau}_{i+1}\tau_{i-1}+\tau_{i+1}\widetilde{\tau}_{i-1})\tau_{i}\widetilde{\tau}_{i}}{(\tau^{2}_{i}-\widetilde{\tau}^{2}_{i})^{2}},\\
\widetilde{\beta}_{i}=2s_{i}s_{i+1}b_{i}\widetilde{b}_{i}=\frac{(\widetilde{\tau}_{i+1}\tau_{i-1}+\tau_{i+1}\widetilde{\tau}_{i-1})(\tau^{2}_{i}+\widetilde{\tau}^{2}_{i})-2(\tau_{i+1}\tau_{i-1}+\widetilde{\tau}_{i+1}\widetilde{\tau}_{i-1})\tau_{i}\widetilde{\tau}_{i}}{(\tau^{2}_{i}-\widetilde{\tau}^{2}_{i})^{2}}.
\end{cases}
\end{equation}
The strongly coupled Toda equations also can be expressed by
\begin{equation}\label{6.19}
\begin{cases}
\tau_{i}\tau_{i}''+\widetilde{\tau}_{i}\widetilde{\tau}_{i}''-(\tau_{i}')^{2}+(\widetilde{\tau}_{i}')^{2}=\tau_{i+1}\tau_{i-1}+\widetilde{\tau}_{i+1}\widetilde{\tau}_{i-1},\\
\widetilde{\tau}_{i}\tau''_{i}+\tau_{i}\widetilde{\tau}''_{i}-2\tau_{i}'\widetilde{\tau}'_{i}=\widetilde{\tau}_{i+1}\tau_{i-1}+\tau_{i+1}\widetilde{\tau}_{i-1}.
\end{cases}
\end{equation}
From \cite{3,9,8}, the $\tau$-functions are written in the form with a simple structure,
\begin{equation}\label{6.20}
\begin{aligned}
\tau_{i}=\sum^{[\frac{i}{2}]}_{q=0}\sum_{\sum\limits_{j=1}^{i}k_{j}=2q}\left|\begin{array}{ccccc}
\tau_{1,[k_{1}]}&\tau^{'}_{1,[k_{2}]}&\cdots&\tau^{(i-1)}_{1,[k_{i-1}]}\\
\tau^{'}_{1,[k_{1}]}&\tau^{''}_{1,[k_{2}]}&\cdots&\tau^{(i)}_{1,[k_{i-1}]}\\
\vdots&\vdots&\ddots&\vdots\\
\tau^{(i-1)}_{1,[k_{1}]}&\tau^{(i)}_{1,[k_{2}]}&\cdots&\tau^{(2i-2)}_{1,[k_{i-1}]}
\end{array}\right|,
\end{aligned}
\end{equation}
\begin{equation}\label{6.21}
\begin{aligned}
\widetilde{\tau}_{i}=\sum^{[\frac{i}{2}]}_{q=0}\sum_{\sum\limits_{j=1}^{i}k_{j}=2q-1}\left|\begin{array}{ccccc}
\tau_{1,[k_{1}]}&\tau^{'}_{1,[k_{2}]}&\cdots&\tau^{(i-1)}_{1,[k_{i-1}]}\\
\tau^{'}_{1,[k_{1}]}&\tau^{''}_{1,[k_{2}]}&\cdots&\tau^{(i)}_{1,[k_{i-1}]}\\
\vdots&\vdots&\ddots&\vdots\\
\tau^{(i-1)}_{1,[k_{1}]}&\tau^{(i)}_{1,[k_{2}]}&\cdots&\tau^{(2i-2)}_{1,[k_{i-1}]}
\end{array}\right|.
\end{aligned}
\end{equation}
Similarly, $\tau^{(i)}_{1,[k_{p}]}=
\begin{cases}
\tau^{(i)}_{1}, k_{p}=0\\
\widetilde{\tau}_{1}^{(i)}, k_{p}=1,
\end{cases}
$ where $\tau_{1}$ is given by $\tau_{1}=c_{11}:=<(\phi^{0}_{1})^{2}+(\widetilde{\phi}^{0}_{1})^{2}>e^{\lambda t}$, $\widetilde{\tau}_{1}=\widetilde{c}_{11}:=<2\phi^{0}_{1}\widetilde{\phi}^{0}_{1}e^{\lambda t}>$, and $\tau^{(i)}_{1}=\frac{d^{i}\tau_{1}}{dt^{i}}$, $\widetilde{\tau}^{(i)}_{1}=\frac{d^{i}\widetilde{\tau}_{1}}{dt^{i}}$. Therefore, the relationships between $\beta_{i}$, $\widetilde{\beta}_{i}$, $D_{i}$ and $\widetilde{D}_{i}$ are given by
\begin{equation}\label{6.22}
\begin{cases}
\tau_{i}=\frac{1}{2s^{i}_{1}}[\prod^{i-1}_{k=1}(\beta^{0}_{k}+\widetilde{\beta}^{0}_{k})^{i-k}(D_{i}+\widetilde{D}_{i})+(\beta^{0}_{k}-\widetilde{\beta}^{0}_{k})^{i-k}(D_{i}-\widetilde{D}_{i})],\\
\widetilde{\tau}_{i}=\frac{1}{2s^{i}_{1}}[\prod^{i-1}_{k=1}(\beta^{0}_{k}+\widetilde{\beta}^{0}_{k})^{i-k}(D_{i}+\widetilde{D}_{i})+(\beta^{0}_{k}-\widetilde{\beta}^{0}_{k})^{i-k}(\widetilde{D}_{i}-D_{i})].
\end{cases}
\end{equation}
\section{Weakly coupled $Z_{n}$-Toda lattices with indefinite metrics}
{\ \ \ \ In the next} part, we will give a new  finite nonperiodic $Z_{n}$-Toda lattices with indefinite metrics as following.
\begin{definition}
According to \eqref{2.3}, we define finite nonperiodic $Z_{n}$-Toda lattice equations with indefinite metrics as:
\begin{equation}\label{5.1}
\begin{cases}
\frac{da_{k,l}}{dt}=s_{k+1}\sum\limits_{{p+q=l+1}}b_{k,p}b_{k,q}-s_{k-1}\sum\limits_{{p+q=l+1}}b_{k-1,p}b_{k-1,q},\\
\frac{db_{k,l}}{dt}=\frac{1}{2}(s_{k+1}\sum\limits_{{p+q=l+1}}b_{k,p}a_{k+1,q}-s_{k}\sum\limits_{{p+q=l+1}}b_{k,p}b_{k,q}).
\end{cases}
\end{equation}
\end{definition}
When $l=1$, $a_{k,l}$ and $b_{k,l}$ are equivalent to $a_{k}$ and $b_{k}$ \cite{1}. In fact, before defining  finite nonperiodic $Z_{n}$-Toda lattice equations, we introduce a more general transformation of variables, the specific transformation is given by
\begin{equation}\label{5.2}
\begin{cases}
s_{k}a_{k,l}=-\frac{y_{k,l}}{2},(k=1,2,\cdots,N)\\
b_{k,l}=\frac{1}{2}\underset{i_{1}k_{1}+\cdots+i_{j}k_{j}=k}\sum\frac{\widehat{x}^{k_{1}}_{i_{1}}\cdots \widehat{x}^{k_{j}}_{i_{j}}}{k_{1}!\cdots k_{p}!}\exp(\frac{x_{k,l}-x_{k+1,l}}{2}), (k=1,2,\cdots,N-1)
\end{cases}
\end{equation}
where $\widehat{x}_{k,l}=\frac{1}{2}(x_{k,l}-x_{k+1,l})$.
Meanwhile, the $Z_{n}$-Hamilton quantity is as
\begin{equation}\label{5.3}
H_{k}=\frac{1}{2}\sum\limits_{k=1}^{N}\sum_{p+q=k}y_{k,p}y_{k,q}+\underset{i_{1}k_{1}+\cdots+i_{j}k_{j}=k}\sum\frac{\overline{x}^{k_{1}}_{i_{1}}\cdots \overline{x}^{k_{j}}_{i_{j}}}{k_{1}!\cdots k_{p}!}\exp(x_{k,1}-x_{k+1,1}),
\end{equation}
where $\overline{x}_{i_{j}}=x_{i,j}-x_{i+1,j}$. According to the definition of the  finite nonperiodic $Z_{n}$-Toda lattice equations with indefinite metrics \eqref{5.1}, we can {obtain} its Lax equations,
\begin{equation}\label{5.4}
\frac{d}{dt}L_{k}=\sum_{p+q=k+1}[B_{p},L_{q}],
\end{equation}
where
\begin{equation}\label{5.5}
L_{k}=\left(\begin{smallmatrix}
s_{1}\sum a_{1,k}\Gamma^{k-1}&s_{2}\sum b_{1,k}\Gamma^{k-1}&0&\cdots&0\\
s_{1}\sum b_{1,k}\Gamma^{k-1}&s_{2}\sum a_{2,k}\Gamma^{k-1}&s_{3}\sum b_{2,k}\Gamma^{k-1}&\cdots&0\\
 &\ddots&\ddots&\vdots\\
0&0&\cdots&s_{N-1}\sum a_{N-1,k}\Gamma^{k-1}&s_{N}\sum b_{N-1,k}\Gamma^{k-1}\\
0&0&\cdots&s_{N-1}\sum b_{N-1,k}\Gamma^{k-1}&s_{N}\sum a_{N,k}\Gamma^{k-1}\\
\end{smallmatrix}
\right)\\,
\end{equation}
and $B_{k}=\frac{1}{2}[(L_{k})_{>0}-(L_{k})_{<0}]$. According to the extended general variables, $\alpha_{k,l}$, $\beta_{k,l}$, $a_{k,l}$ and $b_{k,l}$ are given by
\begin{equation}\label{5.6}
\begin{cases}
\alpha_{k,l}=s_{k}a_{k,l},\\
\beta_{k,l}=\sum_{p+q=k+1}s_{k}s_{k+1}b_{p,l}b_{q,l}.
\end{cases}
\end{equation}
For \eqref{5.4}, linear equations produced from the inverse scattering method can be expressed as:
\begin{equation}\label{5.7}
\begin{cases}
\sum\limits_{p+q=k+1}L_{p}\Phi_{q}=\Lambda\Phi_{k},\\
\frac{d}{dt}\Phi_{k}=\sum\limits_{p+q=k+1}B_{p}\Phi_{q},
\end{cases}
\end{equation}
where $\Phi_{k}$ is the eigenmatrix of $L_{k}$, and $\Phi_{k}\equiv[\phi^{[k]}(\lambda_{1}),\cdots,\phi^{[k]}(\lambda_{N})]=[\phi^{[k]}_{i}(\lambda_{j})]_{1\leq i,j\leq N}$. Further more, $\Phi_{k}$ satisfies
\begin{equation}\label{5.8}
\begin{cases}
\Phi_{1}^{T}S\Phi_{1}=S,\\
\sum\limits_{p+q=k}\Phi_{p}^{T}S\Phi_{q}=0,\ \ k=2,3\cdots,N,\\
\Phi_{1} S^{-1}\Phi_{1}^{T}=S^{-1},\\
\sum\limits_{p+q=k}\Phi_{p} S^{-1}\Phi_{q}^{T}=0,\ \ k=2,3\cdots,N.
\end{cases}
\end{equation}
From \eqref{5.8}, one can get the ``orthogonality'' relations:
\begin{equation}\label{5.10}
\begin{cases}
\sum_{k=1}^{N}s_{k}^{-1}\phi^{[1]}_{i}(\lambda_{k})\phi^{[1]}_{j}(\lambda_{k})=\delta_{ij}s_{i}^{-1},\\
\sum_{k=1}^{N}s_{k}^{-1}\left(\sum\limits^{N}_{p+q=3}\phi^{[p]}_{i}(\lambda_{k})\phi^{[q]}_{j}(\lambda_{k})\right)=0,\\
\sum_{k=1}^{N}s_{k}\phi^{[1]}_{k}(\lambda_{i})\phi^{[1]}_{k}(\lambda_{j})=\delta_{ij}s_{i},\\
\sum_{k=1}^{N}s_{k}\left(\sum\limits^{N}_{p+q=3}\phi^{[p]}_{k}(\lambda_{i})\phi^{[q]}_{k}(\lambda_{j})\right)=0.
\end{cases}
\end{equation}
So, the $L_{k}$ can be expressed by
\begin{equation}\label{5.11}
L_{k}=\sum\limits_{p+q=k+1}\Phi_{p}\Lambda\Phi^{T}_{q}.
\end{equation}
According to the proof of section 3, the specific form of $\phi^{[k]}_{i}(\lambda,t)$  can be given as
\begin{align}\label{5.12}
\phi^{[k]}_{i}(\lambda,t)&=\overline{D}_{0}(t)\sum\limits_{j=1}^{N-1}\left(\sum\limits_{[p_{1}+\cdots+p_{i-1}-(i-1)]=j}\left|\begin{array}{cccccc}
s_{1}c^{[k_{p_{1}}]}_{11}&\cdots&s_{i-1}c^{[k_{p_{i-1}}]}_{1,i-1}&\phi^{0}_{1}(\lambda)\\
\vdots&\vdots&\ddots\\
s_{1}c^{[k_{p_{1}}]}_{i1}&\cdots&s_{i-1}c^{[k_{p_{i-1}}]}_{i,i-1}&\phi^{0}_{i}(\lambda)
\end{array}\right|\right)\\
&+\overline{D}_{1}(t)\sum\limits_{j=1}^{N-1}\left(\sum\limits_{[p_{1}+\cdots+p_{i-1}-(i-1)]=j}\left|\begin{array}{cccccc}
s_{1}c^{[k_{p_{1}}]}_{11}&\cdots&s_{i-1}c^{[k_{p_{i-1}}]}_{1,i-1}&\phi^{0}_{1}(\lambda)\\
\vdots&\vdots&\ddots\\
s_{1}c^{[k_{p_{1}}]}_{i1}&\cdots&s_{i-1}c^{[k_{p_{i-1}}]}_{i,i-1}&\phi^{0}_{i}(\lambda)
\end{array}\right|\right)\\
&+\cdots+\overline{D}_{n-1}(t)\left|\begin{array}{cccccc}
s_{1}c^{[k_{1}]}_{11}&\cdots&s_{i-1}c^{[k_{i-1}]}_{1,i-1}&\phi^{0}_{1}(\lambda)\\
\vdots&\vdots&\ddots\\
s_{1}c^{[k_{1}]}_{i1}&\cdots&s_{i-1}c^{[k_{i-1}]}_{i,i-1}&\phi^{0}_{i}(\lambda)
\end{array}\right|,
\end{align}
where $\overline{D}_{i}(t)$ is the $n$-th order Frobenius form of $\frac{e^{\lambda t}}{\sqrt{D_{i}(t)D_{i-1}(t)}}$, and the specific forms are given by $\overline{D}_{i}(t)=\underset{k_{0}+k_{1}+\cdots+k_{p}=i-1}\sum\frac{(-1)^{i}v_{0}^{k_{0}}v_{1}^{k_{1}}\cdots v_{i-1}^{k_{i-1}}}{v^{i}_{0}}e^{\lambda t}$. For $v_{i}$ above, we used a variable replacement: $u_{i}=\underset{p+q=i+1}\sum D_{i,p}(t)D_{i-1,q}(t)$, so we can get the relationship between $v_{k}$ and $u_{k}$ by iterative methods.
\\

{\bf {Acknowledgements:}}
Chuanzhong Li  is  supported by the National Natural Science Foundation of China under Grant No. 11571192 and K. C. Wong Magna Fund in
Ningbo University.

\end{document}